%% file: main-ide-revised2.tex
\documentclass[final, 3p, authoryear, 10pt]{elsarticle} 
\journal{Ecological Modelling} 
\usepackage[nodots]{numcompress}
\usepackage{amsmath,amssymb,amsthm,amsfonts} 
\usepackage[]{graphicx}
\usepackage{lineno,setspace} 
\usepackage[pdftitle={Equivalence of the realized input and output
  oriented indirect effects metrics in ecological network analysis},
pdfauthor={Stuart R. Borrett, Michael A. Freeze, Andria K. Salas},
pdfkeywords={ecology, network, indirect effects, environ, NEA,
  holoecology}, pdfpagemode=UseOutlines, bookmarks=T, colorlinks,
linkcolor=blue,citecolor=blue, pdfstartview=FitH,
urlcolor=red]{hyperref}

\def\tableline{ \vskip .1in \hrule height .6pt \vskip 0.1in}

\newcommand{\ID}{I/D}

\def\citeapos#1{\citeauthor{#1}'s (\citeyear{#1})}

\theoremstyle{plain}
\newtheorem{theorem}{Theorem}[section]
\newtheorem{proposition}[theorem]{Proposition}
\newtheorem{lemma}[theorem]{Lemma}

\theoremstyle{definition}

\begin{document} 
\begin{frontmatter}
  \title{Equivalence of the realized input and output oriented
    indirect effects metrics in ecological network analysis}

\author[uncwbio,cms]{S.R.~Borrett\corref{cor1}}
\ead{borretts@uncw.edu}

\author[uncwmath]{M.A.~Freeze}
\ead{freezem@uncw.edu}

\author[uncwbio,cms]{A.K.~Salas}
\ead{aks2515@gmail.com}

\address[uncwbio]{Department of Biology \& Marine Biology, University
  of North Carolina Wilmington, 601 S.\ College Rd., Wilmington, 28403
  NC, USA}
\address[cms]{Center for Marine Science, University of North Carolina Wilmington}
\address[uncwmath]{Department of Math \& Statistics, University
  of North Carolina Wilmington, 601 S.\ College Rd., Wilmington, NC
  28403, USA}

\cortext[cor1]{Corresponding author. Tel. 910.962.2411; fax: 910.962.4066}

\begin{abstract}
  A new understanding of the consequences of how ecosystem elements
  are interconnected is emerging from the development and application
  of Ecological Network Analysis.  The relative importance of indirect
  effects is central to this understanding, and the ratio of indirect
  flow to direct flow (\ID) is one indicator of their importance.  Two
  methods have been proposed for calculating this indicator.  The unit
  approach shows what would happen if each system member had a unit
  input or output, while the realized technique determines the ratio
  using the observed system inputs or outputs.  When using the unit
  method, the input oriented and output oriented ratios can be
  different, potentially leading to conflicting results.  However, we
  show that the input and output oriented \ID\ ratios are identical
  using the realized method when the system is at steady state.  This
  work is a step in the maturation of Ecological Network Analysis that
  will let it be more readily testable empirically and ultimately more
  useful for environmental assessment and management.
\end{abstract}


\end{frontmatter}


\section{Introduction}
The dominance of indirect effects hypothesis is central to our new
understanding of ecosystem ecology revealed through systems analyses
like Ecological Network Analysis (ENA) \citep{ulanowicz86,
  fath99_review, dunne02food, belgrano05, bascompte07,
  jorgensen07_newecology, mcrae08circuit, christian2009ecological,
  baird2011}.  While empirical results are also highlighting the
important roles of indirect effects \citep[e.g.][]{wootton94_review,
  berger08, diekotter07, menendez07, letnic2009keystone, walsh10},
network analysis shows how the integral of indirect interactions
across a whole web of interactions can transform the effective
relationships between species so that they tend to be more positive
\citep{ulanowicz90, patten91, fath98, bondavalli99}, and more evenly
distribute the system resources \citep{fath99_homo, borrett10_hmg} and
control \citep{schramski06, schramski07}.

When applying ENA, ecologists use the ratio of indirect to direct
flows (\ID) to indicate the importance of indirect effects
\citep{higashi89, fath99_review, jorgensen07_newecology}.  For
example, \citet{borrett06_neuse} used \ID\ to show the temporal
consistency of the organization of the Neuse River Estuary ecosystem,
and \citet{baird09_agg} found that indirect effects were dominant in
the Sylt-R{\o}mo Bight ecosystem regardless of the model aggregation
scheme.  \citet{borrett11_ree} distinguished between two ways of
calculating this ratio that they termed the \emph{unit} and
\emph{realized} formulations.  The unit calculation is the traditional
approach that is performed after a classic systems analysis technique
is applied.  In this case, each species is assumed to receive a single
unit of boundary flux, and then the direct and indirect flow
intensities in the system are summed.  This unit calculation shows
what the relative indirect flows would be if every node received the
same input.  It is useful because it facilitates within system
comparisons \citep{whipple07} and due to its close relationship to the
eigenvalues and vectors of underlying matrices \citep{borrett10_idd}.
Alternatively, the realized calculation uses the observed boundary
flow vector instead of an input vector of ones.  This approach was
first introduced by \citet{borrett06_neuse} but was clarified and
formalized in \citet{borrett11_ree}.  Its analytical advantage is that
the system flows are scaled with the observed, typically
non-homogeneous boundary fluxes.  

\citet{borrett11_ree} showed that the quantitative and qualitative
results could be changed by this shift in perspective, highlighting
the importance of the system's connection to its broader environment.
For example, the wet season model of carbon flux in the Everglades
graminoid marshes \citep{ulanowicz00_graminoids, heymans02} exhibits
the dominance of indirect effects when analyzed with the realized
method (\ID\ = 1.2) but does not when analyzed with the unit approach
(\ID = 0.81).  In this example, both the numerical value of the
indicator and its interpretation changed.

An important aspect of ENA throughflow analysis is that it has two
orientations: input and output.  Conceptually, we can either pull the
energy--matter out of the system and trace its origin back through the
system to the inputs, or we can push the inputs into the system and
then trace where they travel through the system until they exit the
system.  The first type of analysis is termed the input analysis
because in \citeapos{leontief66} original application he was
determining what raw materials were required as inputs into the
economic system to generate a single output, like a car. The second
analysis is forward looking in direction, but it is termed the output
analysis because we are following the inputs to their output from the
system.  In this case, it is the output generated that is the focus.
Patten exploited this bidirectional analytical feature in the
development of his environ concept \citep{patten78,
  patten81_superniche, patten82} and the subsequent environmental
theory.  While some ENA components leverage the differences between
the input and output orientation \citep[e.g.][]{patten81,
  ulanowicz90,fath98,patten91, schramski06, schramski07}, the
alternative perspectives can generate conflicting and confusing
results.

In this short communication we show that while the input and output
oriented unit \ID\ can differ in quantitative and qualitative nature,
the realized input and output oriented \ID\ are always identical.  We
first show this result numerically for a set of network ecosystem
models drawn from the literature.  We then provide a mathematical
proof for why this relationship will generally hold when the models
are at steady state.  We conclude this paper by arguing that this
identity bolsters the utility of the realized metric for (1) testing
the generality of network hypotheses like the dominance of indirect
effects \citep{salas11_did}, (2) applying the ecological insights of
this type of ENA to the original system, and (3) comparing ecosystems'
organization \citep{baird91, baird09_agg, borrett06_neuse,
  borrett07_lanier, whipple07}.

\section{Ecological Network Analysis}
ENA is well described in the literature \citep[e.g.,][]{patten76,
  ulanowicz86, ulanowicz04, fath99_review, fath06, schramski11}, but
we recount the input and output throughflow analyses here as this is
essential for our discussion.  \citet{borrett11_ree} focused on the
distinction between the idealized and realized forms of the \ID\
ratio.  Here we focus on contrasting the input and output oriented
analyses.

\subsection{Model Input}
Ecologists apply ENA to network models of energy--matter storage and
flux in ecosystems.  These models are alternatively referred to as
compartment models or energy--matter budgets.  In the general form of
these models, $n$ nodes represent species, groups of species, or
abiotic components, and the $L$ weighted directed edges represent the
flow of energy--matter generated by some ecological process (e.g.,
consumption, excretion, harvesting).  Let $\mathbf{F}_{n\times
  n}=(f_{ij})$ represent the observed flow from ecosystem compartment
$j$ to compartment $i$ (e.g., $j \rightarrow i$), $\vec{z}_{n\times
  1}$ be a vector of node inputs originating from outside the system,
and $\vec{y}_{1\times n}$ be a vector of flows from each node that
exit the system. Sometimes the outputs are subdivided into those
that are thermodynamically degraded and those that reflect high
quality energy--matter exported from the system
\citep[e.g.,][]{ulanowicz86}.  For the analysis presented here, we
lump these losses together because once they cross the system boundary
our analysis is blind to their fate, making this distinction less
critical.

ENA throughflow analysis assumes that the model traces a single
conserved currency (e.g., energy, nitrogen, phosphorus), and that the
system is at a steady state (inputs equal outputs).  The steady state
assumption is crucial for the throughflow path decomposition
\citep{borrett10_idd}.

For this research, we examined 50 ecosystem network models that
represent 35 distinct primarily freshwater and marine ecosystems
(Table~\ref{tab:models}).  The models exhibit a range of sizes ($4 \le
n \le 125$), connectance (number of direct links divided by the total
possible number of links; $0.03 \le (C=L/n^2) \le 0.40$), and
recycling (Finn Cycling Index \citep{finn80}, which is the proportion
of total system throughflow derived from recycled flux; $0 \le FCI \le
0.51$).  The core of these trophically-based networks is a food web,
but other ecological processes are included and model compartments may
represent non-living resource pools such as particulate organic
carbon.  We applied \citeapos{allesina03} AVG2 balancing algorithm to
15 models that were not initially at steady state because this
technique tends to generate the least distortion of network
properties. This is the same set of models used in several recent
network studies \citep{borrett10_hmg, borrett11_ree, salas11_did}, and
the model data are available from \href{http://people.uncw.edu/borretts/research.html}{http://people.uncw.edu/borretts/research.html}.

\subsection{Throughflow Analysis}
Input and output oriented throughflow analysis have three main steps.
First, we determine the throughflow vector $\vec{T}$, which is the
total amount of energy--matter flowing into or out of each node.  This
can be calculated from the initial model information as follows:
\begin{linenomath}
\begin{align}
T_i^{\textrm{in}}&\equiv \sum_{j=1}^nf_{ij} + z_i \quad (i = 1, 2, \ldots,
n) \textrm{, and}\\
T_j^{\textrm{out}}&\equiv \sum_{i=1}^nf_{ij} + y_j \quad (j = 1, 2, \ldots,
n).
\end{align}
\end{linenomath}
At steady state, $T_i^{\textrm{ in}} = (T_j^{\textrm{
    out}})^{\mathrm{T}} = \vec{T}_{n\times 1} = (T_j)$.  From this
vector, we derive the first whole-system indicator, total system
throughflow ($TST=\sum_{j=1}^nT_j$). $TST$ indicates the total magnitude
of flow activity. 

Second, we calculate the input $\mathbf{G'}_{n\times n}=(g'_{ij})$ and
output $\mathbf{G}_{n\times n}=(g_{ij})$ \emph{direct flow
  intensities} from node $j$ to $i$.  These are defined as
\begin{linenomath}
\begin{align} 
g'_{ij}&\equiv  f_{ij}/T^{\textrm{in}}_i \textrm{, and} \label{eq:gij}\\
g_{ij}&\equiv   f_{ij}/T^{\textrm{out}}_j. \label{eq:gpij}
\end{align}
\end{linenomath}
Here, $g'_{ij}$ is the fraction of input at receiver node $i$
contributed from the donor node $j$, while $g_{ij}$ is the fraction of
output throughflow at donor node $j$ contributed to node $i$.  The
subtle distinction between the input and output oriented direct flow
intensities in equations (\ref{eq:gij} and \ref{eq:gpij}) is whether
the observed flow is normalized by the throughflow of the receiving
node (input case) or the donating node (output case).  The $g'_{ij}$
and $g_{ij}$ values are dimensionless and the column sums of the
matrices must lie between 0 and 1 because of thermodynamic constraints
of the original model \citep[see][for details]{borrett10_idd}.

Third, we determine the \emph{integral flow intensities}
$\mathbf{N'}=(n_{ij})$ (input) and $\mathbf{N}=(n_{ij})$ (output) as
\begin{linenomath}
\begin{align} 
  \mathbf{N'} &\equiv \sum_{m=0}^\infty \mathbf{G'}^m =
  \underbrace{\mathbf{I}}_{\textrm{Boundary}} +
  \underbrace{\mathbf{G'}^1}_{\textrm{Direct}} +
  \underbrace{\mathbf{G'}^2 + \ldots + \mathbf{G'}^m +
    \ldots}_{\textrm{Indirect}}\textrm{,} \label{eq:Ga} \\
\textrm{and} \nonumber\\  
  \mathbf{N} &\equiv \sum_{m=0}^\infty \mathbf{G}^m = 
  \underbrace{\mathbf{I}}_{\textrm{Boundary}} +
  \underbrace{\mathbf{G}^1}_{\textrm{Direct}} + \underbrace{\mathbf{G}^2
    + \ldots + \mathbf{G}^m +
    \ldots}_{\textrm{Indirect}}.  \label{eq:Gb} 
\end{align}
\end{linenomath}
In equations (\ref{eq:Ga} and \ref{eq:Gb}), $\mathbf{I}=(i_{ij})$ is
the matrix multiplicative identity and the elements of $\mathbf{G'}^m$
and $\mathbf{G}^m$ are the fractions of boundary flow that travels
from node $j$ to $i$ over all pathways of length $m$.  The exact
values of $\mathbf{N'}$ and $\mathbf{N}$ can be found using the
following identities because given our model definitions the power
series must converge.  Thus,
\begin{linenomath}
\begin{align} \mathbf{N'} &= (\mathbf{I}-\mathbf{G'})^{-1} \textrm{, and} \\
\mathbf{N} &= (\mathbf{I}-\mathbf{G})^{-1}. 
\end{align}
\end{linenomath}
The $n'_{ij}$ and $n_{ij}$ elements represent the intensity of
boundary flow that passes from $j$ to $i$ over all pathways of all
lengths.  These values integrate the boundary, direct, and indirect
flows.  

The differences between $\mathbf{N'}$ and $\mathbf{N}$ results
from the input and output perspectives.  However, we can use either of
these integral flow matrices to recover $T$ as follows:
\begin{linenomath}
\begin{align}
\vec{T}&=\vec{y}\mathbf{N'}\textrm{, } \label{eq:TNP}\\ 
\vec{T}&=\mathbf{N}\vec{z}.           \label{eq:TN}
\end{align}
\end{linenomath}

\subsection{Unit and Realized Indirect Effects}

There are two methods of calculating \ID\ used in ENA to quantify the
importance of indirect flows.  We refer to these as the \emph{unit}
and \emph{realized} methods.  The unit method assumes that each node
receives a single unit of input, which we will represent as a vector
of ones with length equal to the number of nodes $[1]_{n \times 1}$ or
output $[1]_{1 \times n}$.
The unit input and output \ID\ metrics are calculated as follows.
\begin{linenomath}
\begin{align}
  \ID_{\textrm{unit, input}} &= \frac{\sum_{i=1}^n\bigl(\vec{[1]}_{1\times
      n}\left(\mathbf{N'}-\mathbf{I}-\mathbf{\mathbf{G}'}\right)\bigr)_i}{\sum_{i=1}^n\bigl(\vec{[1]}_{1\times
      n}\mathbf{\mathbf{G}'}\bigr)_i} \\
  \ID_{\textrm{unit, output}} &= \frac{\sum_{i=1}^n\bigl(\left(\mathbf{N}-\mathbf{I}-\mathbf{G}\right)\vec{[1]}_{n\times 1}\bigr)_i}{\sum_{i=1}^n\bigl(\mathbf{G}\vec{[1]}_{n\times 1}\bigr)_i}
\end{align}
\end{linenomath}

In contrast, the realized method uses the observed boundary vectors
$\vec{z}$ or $\vec{y}$.  
\begin{linenomath}
\begin{align}
\ID_{\textrm{realized, input}}&= \frac{\sum_{i=1}^n\bigl(\vec{y}(\mathbf{N'}-\mathbf{I}-\mathbf{G'})\bigr)_i}{\sum_{i=1}^n\bigl(\vec{y}\mathbf{G'}\bigr)_i}
\\
\ID_{\textrm{realized, output}}&= \frac{\sum_{i=1}^n\bigl((\mathbf{N}-\mathbf{I}-\mathbf{G})\vec{z}\bigr)_i}{\sum_{i=1}^n\bigl(\mathbf{G}\vec{z}\bigr)_i}
\end{align}
\end{linenomath}

As \citet{borrett11_ree} discuss, the magnitude of the boundary flow
vectors, be it a unit vector $\vec{[1]}_{n \times 1}$ or the observed
boundary vectors $\vec{z}$ or $\vec{y}$, cancels in the ratio measure.
However, the different distribution of boundary flows differentially
excites the flow intensities in the integral flow matricies.  It is
the distribution of the boundary inputs and outputs that is key.

\section{Results}
We present our results in two parts.  The first section shows the
numerical results of applying ENA and the alternative formulations of
\ID\ to the 50 network models.  The second section summarizes the
analytical work that generalizes this result.

\subsection{Numerical Results}
Figure~\ref{fig:ide} compares the input and output oriented \ID\
values of the network models.  Panel (a) illustrates the variation
that can occur between the input and output oriented \ID\ values when
using the idealized formulation.  The adjusted $r^2$ value for a
linear regression fit to this data set is $0.82$, and the Pearson's
correlation coefficient is $0.91$.  To ensure our balancing routine
did not heavily influence our results, we also calculated the
Pearson's correlation coefficient for only the 35 models not balanced
and found it to be $0.92$.
Panel (b) shows the one to one correspondence between the input and
output \ID\ values when we use the realized formulation.  These values
are identical and have a Pearson's correlation coefficient of 1.  
 

\subsection{Analytical Results}

The steady state assumption actually implies very strong compatiblity
conditions on the powers of $\mathbf{G}$ and $\mathbf{G'}$, namely
that for each nonnegative integer $m$, the sizes of $\mathbf{G}^{m}\vec{z}$ and
$\vec{y}\mathbf{G}^{' m}$ must coincide in the sense that we have the
equality
\begin{linenomath}
\begin{align}
\sum_{i=1}^n (\mathbf{G}^{m}\vec{z})_{i} &= \sum_{i=1}^n (\vec{y}\mathbf{G}^{' m})_{i}, \label{equalterms}
\end{align}
\end{linenomath}
as demonstrated in the proof of Lemma \ref{lemmatwo} in Appendix A.

The terms of the numerator and denominator of realized \ID\ each are
of the form found on the left hand side in equation~(\ref{equalterms})
for realized output \ID\ or of the form found on the right hand side
in equation~(\ref{equalterms}) for realized input \ID\ .  It follows
that at steady state, the realized input \ID\ and output \ID\ must be
identical.  This mathematical generalization of the numerically
observed results is recorded in Theorem \ref{thetheorem} of Appendix
A.


\section{Discussion}
In this note we show that using the realized formulation the input and
output oriented \ID\ values are identical when the system is at steady
state.  Like \citet{bata07}, this work provides a conceptual and
mathematical condensation for Ecological Network Analysis.  This is
part of a methodological maturation process in which investigators are
finding relationships amongst the alternative analyses (input vs.\
output or throughflow vs.\ storage) and reconsidering initial
assumptions that have become canalized \citep[e.g.][]{schramski11,
  matamba09}.  Here, we are not fully pruning the analytical options,
but clarifying the consequences of alternative methods.  When an
investigator selects to use the realized method as described by
\citet{borrett11_ree}, one advantage is that the input and output
values are identical.

The realized formulation of \ID\ functions as one description of the
internal organization of the ecosystem.  It is intuitively satisfying
that our understanding of this organization does not change because
our point of view switches (input or output).  Further, as the
realized formulation considers both the internal system organization
and the external environmental connection, it characterizes the
complete system as it was initially described by the model builders.
Thus, the realized metric is most appropriate for testing the
generality of systems ecology hypotheses like the dominance of
indirect effects \citep{salas11_did}, applying ENA insights back to
the original system \citep{zhang10, zhang10_ecomod}, and comparing
ecosystem organization \citep{baird91, whipple07, ray08}.

A limitation of this work is that the realized input and output \ID\
should only be identical when the observed system is at steady state,
as shown in the mathematical proof.  This means systems not at steady
state will likely have different relative contributions of indirect
flows when observed forward through the network instead of backwards
through the network.  This is evident when we recall equations
(\ref{eq:TNP}) and (\ref{eq:TN}) and that $T^{in} = T^{out}$ at steady
state.  Without this identity (i.e., not at steady state), there is no
reason to expect the mathematical decomposition of the input and
output throughflow into boundary, direct, and indirect flows to be
identical.  We suspect that network particle tracking as implemented
in EcoNet \citep{kazanci07, tollner09} will be useful for further
exploring the consequences of these input and output perspective
differences.

This work is one more essential step in the maturation of the
ecosystem and environmental theory developing in association with
Ecosystem Network Analysis \citep[e.g.][]{abarca02, allesina04_cycles,
  borrett10_idd, kaufman10, schramski11}.  Multiple lines of
investigation are building confidence in ENA results including the
initial empirical validation by \citet{dame08_validation} and the
agreement between the Eulerian approach used here and the Lagrangian
approach of network particle tracking \citep{matamba09}, but continued
rigorous empirical testing of network ecology hypotheses is essential.
\citet{loehle87_hypothesis} points out that this theory maturation is
essential so that theory predictions clarify to the point that robust
empirical test be successful.

\section{Acknowledgments}
This paper benefited from discussions with B.C. Patten and the careful
reviews J. R. Schramski, A. Stapleton, and two anonymous reviewers.
This work was supported in part by UNCW and NSF (DEB-1020944).  AKS
was supported by the James F.\ Merritt fellowship from the UNCW Center
for the Marine Science.

\section{Appendix A}

The steady state assumption that $\vec{T}^{\textrm{in}} =
(\vec{T}^{\textrm{out}})^{\mathrm{T}}$ means that
$T_{\ell}^{\textrm{in}} = T_{\ell}^{\textrm{out}}$ for all indices
$\ell$.  We have the following proposition as an immediate
consequence.


\begin{proposition}\label{propm01}
  When the observed system is at steady state, we have that 
  $\sum_{i} (\mathbf{G}\vec{z})_{i} = \sum_{i} (\vec{y}\mathbf{G}')_{i}.$
\end{proposition}

\begin{proof}



Note that
\begin{align*}
\sum_{i} (\mathbf{G}\vec{z})_{i} &= \sum_{i} \left[ \sum_{j} \left( z_{j} \frac{ f_{ij} }{T_{j}^{\textrm{out}}} \right) \right]\\
&= \sum_{i,j} \left( z_{j} \frac{ f_{ij} }{T_{j}^{\textrm{in}}} \right) \\
&= \sum_{i,j} \left[  \left( T_{j}^{\textrm{in}} - \sum_{k} (f_{jk}) \right) \frac{ f_{ij} }{T_{j}^{\textrm{in}}} \right] (\textrm{by equation (1)}) \\
&= \sum_{i,j} (f_{ij}) - \sum_{i,j,k} \left( \frac{ f_{ij} f_{jk} }{T_{j}^{\textrm{in}}} \right) \\
&= \sum_{i,j} (f_{ij}) - \sum_{j,k} \left[ \frac{ f_{jk} }{ T_{j}^{\textrm{in}} } \sum_{i} ( f_{ij} ) \right] \\
&= \sum_{i,j} (f_{ij}) - \sum_{j,k} \left[ \frac{ f_{jk} }{ T_{j}^{\textrm{in}} } \left( T_{j}^{\textrm{out}} - y_{j} \right) \right] (\textrm{by equation (2)}) \\
&= \sum_{i,j} (f_{ij}) - \sum_{j,k} (f_{jk}) + \sum_{j,k} \left( \frac{ f_{jk} }{ T_{j}^{\textrm{out}} } y_{j} \right) \\
&=  \sum_{j,k} \left( \frac{ f_{jk} }{ T_{j}^{\textrm{out}} } y_{j} \right) \\
&= \sum_{k} \left[ \sum_{j}  \left( y_{j} \frac{ f_{jk} }{ T_{j}^{\textrm{out}} }  \right) \right] \\
&= \sum_{k} (\vec{y}\mathbf{G}')_{k}.
\end{align*}

\end{proof}

We extend the observation of Proposition \ref{propm01}.

\begin{lemma}\label{lemmaone}
For each integer $m \ge 2$, we have that
$$(\mathbf{G}^{m}\vec{z})_i = \sum_{\ell_{1},\dots,\ell_{m}} \left[ z_{\ell_{m}} \cdot \frac{ f_{i \ell_{1}} }{T_{\ell_{1}}^{\textrm{out}}} \cdot \prod_{n=1}^{m-1} \left( \frac{ f_{\ell_{n} \ell_{n+1}}}{ T_{\ell_{n+1}}^{\textrm{out}}} \right) \right]$$
and
$$(\vec{y}\mathbf{G}^{' m})_i = \sum_{\ell_{1},\dots,\ell_{m}} \left[ y_{\ell_{m}} \cdot \frac{ f_{\ell_{1} i} }{ T_{\ell_{1}}^{\textrm{in}}} \cdot \prod_{n=1}^{m-1} \left( \frac{ f_{\ell_{n+1} \ell_{n}}}{ T_{\ell_{n+1}}^{\textrm{in}}} \right) \right].$$
\end{lemma}

\begin{proof}
We observed in the proof of Proposition \ref{propm01} that
$$(\mathbf{G}\vec{z})_{i} = \sum_{j} \left( z_{j} \frac{ f_{ij} }{T_{j}^{\textrm{out}}} \right).$$
It follows that
\begin{align*}
(\mathbf{G}^{2}\vec{z})_{i} &= (\mathbf{G}(\mathbf{G}\vec{z}))_{i} \\
&= \sum_{j} \left( (\mathbf{G}\vec{z})_{j} \frac{ f_{ij} }{T_{j}^{\textrm{out}}} \right) \\
&= \sum_{j} \left[ \left( \sum_{k} z_{k} \frac{ f_{jk} }{T_{k}^{\textrm{out}}} \right)  \frac{ f_{ij} }{T_{j}^{\textrm{out}}}  \right] \\
&= \sum_{j,k} \left( z_{k} \frac{ f_{ij} }{T_{j}^{\textrm{out}}} \frac{ f_{jk} }{T_{k}^{\textrm{out}}} \right) \\
&= \sum_{\ell_{1},\ell_{2}} \left( z_{\ell_{2}} \frac{ f_{i \ell_{1}} }{T_{\ell_{1}}^{\textrm{out}}} \frac{ f_{\ell_{1} \ell_{2}} }{T_{\ell_{2}}^{\textrm{out}}} \right),
\end{align*}
and by iteration we obtain the stated formula for $(\mathbf{G}^{m}\vec{z})_i .$  The formula for $(\vec{y}\mathbf{G}^{' m})_i$ is similarly obtained.
\end{proof}

The following lemma makes use of Lemma \ref{lemmaone} to extend the observations of Proposition \ref{propm01}.

\begin{lemma}\label{lemmatwo}
For each integer $m \ge 0$, we have that
$$\sum_{i=1}^n (\mathbf{G}^{m}\vec{z})_{i} = \sum_{i=1}^n (\vec{y}\mathbf{G}^{' m})_{i}.$$
\end{lemma}

\begin{proof}


The case $m=0$, namely that $\sum_{i} z_{i} = \sum_{i} y_{i}$, follows directly from the steady-state assumption, and
the case $m=1$ was observed in Proposition \ref{propm01}.  Suppose then that $m \ge 2$.

Observe that we have the following identity by re-indexing the left-hand side via $\ell_{t} \mapsto \ell_{t+1}$ for $1 \le t \le m-1$, $\ell_{m} \mapsto \ell_{1}$, followed by renaming $\ell_{1}$ to $\ell_{m+1}$ and then $i$ to $\ell_{1}$:
 \begin{align}
\sum_{i, \ell_{1},\dots,\ell_{m}} \left[ f_{i \ell_{1}} \cdot \prod_{n=1}^{m-1} \left( \frac{ f_{\ell_{n} \ell_{n+1}}}{ T_{\ell_{n}}^{\textrm{out}}} \right)    \right] &=
 \sum_{\ell_{1}, \ldots, \ell_{m+1}}  \left[ f_{\ell_{1} \ell_{2}} \prod_{n=2}^{m} \left( \frac{ f_{\ell_{n} \ell_{n+1}}}{ T_{\ell_{n}}^{\textrm{out}}} \right) \right] \label{initialobservation}
\end{align}

From Lemma \ref{lemmaone} it follows that
\begin{align*}
\sum_{i} (\mathbf{G}^{m} \vec{z})_{i} &= \sum_{i, \ell_{1}, \ldots, \ell_{m}} \left[ z_{\ell_{m}} \cdot \frac{  f_{i \ell_{1}}}{ T_{\ell_{1}}^{\textrm{out}} } \cdot \prod_{n=1}^{m-1} \left( \frac{ f_{\ell_{n} \ell_{n+1}}}{T_{\ell_{n+1}}^{\textrm{out}}} \right) \right] \\
&=  \sum_{i, \ell_{1},\dots,\ell_{m}} \left[ \left( T_{\ell_{m}}^{\textrm{out}} - \sum_{\ell_{m+1}} f_{\ell_{m} \ell_{m+1}} \right) \cdot \frac{  f_{i \ell_{1}}}{ T_{\ell_{1}}^{\textrm{out}} } \cdot \prod_{n=1}^{m-1} \left( \frac{ f_{\ell_{n} \ell_{n+1}}}{T_{\ell_{n+1}}^{\textrm{out}}} \right) \right] (\textrm{recall that } \vec{T}^{\textrm{in}} = \vec{T}^{\textrm{out}}) \\
&= \sum_{i, \ell_{1},\dots,\ell_{m}} \left[ f_{i \ell_{1}} \cdot \prod_{n=1}^{m-1} \left( \frac{ f_{\ell_{n} \ell_{n+1}}}{T_{\ell_{n}}^{\textrm{out}}} \right)    \right] - \sum_{i, \ell_{1}, \ldots, \ell_{m+1}}  \left[ f_{i \ell_{1}} \prod_{n=1}^{m} \left( \frac{ f_{\ell_{n} \ell_{n+1}}}{T_{\ell_{n}}^{\textrm{out}}} \right) \right] \\
&=  \sum_{i, \ell_{1},\dots,\ell_{m}} \left[ f_{i \ell_{1}} \cdot \prod_{n=1}^{m-1} \left( \frac{ f_{\ell_{n} \ell_{n+1}}}{T_{\ell_{n}}^{\textrm{out}}} \right)    \right] - \sum_{\ell_{1}, \ldots, \ell_{m+1}}  \left[ \left( \sum_{i} f_{i \ell_{1}} \right) \prod_{n=1}^{m} \left( \frac{ f_{\ell_{n} \ell_{n+1}}}{T_{\ell_{n}}^{\textrm{out}}} \right) \right] \\
&=  \sum_{i, \ell_{1},\dots,\ell_{m}} \left[ f_{i \ell_{1}} \cdot \prod_{n=1}^{m-1} \left( \frac{ f_{\ell_{n} \ell_{n+1}}}{T_{\ell_{n}}^{\textrm{out}}} \right)    \right] - \sum_{\ell_{1}, \ldots, \ell_{m+1}}  \left[ \left(T_{\ell_{1}}^{\textrm{out}} - y_{\ell_{1}} \right) \prod_{n=1}^{m} \left( \frac{ f_{\ell_{n} \ell_{n+1}}}{T_{\ell_{n}}^{\textrm{out}}} \right) \right] \\
&=  \sum_{i, \ell_{1},\dots,\ell_{m}} \left[ f_{i \ell_{1}} \cdot \prod_{n=1}^{m-1} \left( \frac{ f_{\ell_{n} \ell_{n+1}}}{T_{\ell_{n}}^{\textrm{out}}} \right)    \right] -
 \sum_{\ell_{1}, \ldots, \ell_{m+1}}  \left[ f_{\ell_{1} \ell_{2}} \prod_{n=2}^{m} \left( \frac{ f_{\ell_{n} \ell_{n+1}}}{T_{\ell_{n}}^{\textrm{out}}} \right) \right] +
 \sum_{\ell_{1}, \ldots, \ell_{m+1}}  \left[ y_{\ell_{1}} \cdot \prod_{n=1}^{m} \left( \frac{ f_{\ell_{n} \ell_{n+1}}}{T_{\ell_{n}}^{\textrm{out}}} \right) \right] \\
&= \sum_{\ell_{1}, \ldots, \ell_{m+1}}  \left[ y_{\ell_{1}} \cdot \prod_{n=1}^{m} \left( \frac{ f_{\ell_{n} \ell_{n+1}}}{T_{\ell_{n}}^{\textrm{out}}} \right) \right] (\textrm{by equation (\ref{initialobservation})}) \\
&=  \sum_{i, \ell_{1},\dots,\ell_{m}} \left[ y_{\ell_{m}} \cdot \frac{ f_{\ell_{1} i} }{T_{\ell_{1}}^{\textrm{in}}} \cdot \prod_{n=1}^{m-1} \left( \frac{ f_{\ell_{n+1} \ell_{n}}}{T_{\ell_{n+1}}^{\textrm{in}}} \right) \right]  (\textrm{relabelling } \ell_{m+1} \textrm{ as } i) \\
&= \sum_{i} ( \vec{y} \mathbf{G}^{' m})_{i} .
\end{align*}

\end{proof}

\begin{theorem}\label{thetheorem}
Realized input and output \ID\ are identical when the observed system is at steady state.
\end{theorem}

\begin{proof}
  Lemma \ref{lemmatwo} implies the equality of the denominators of
  realized input and output \ID\ as well as equality of corresponding
  summands of the numerators of realized input and output \ID\ .
\end{proof}



\clearpage

\input{table_models_ide}


\begin{figure}[t]
 \center
 \includegraphics[scale=1]{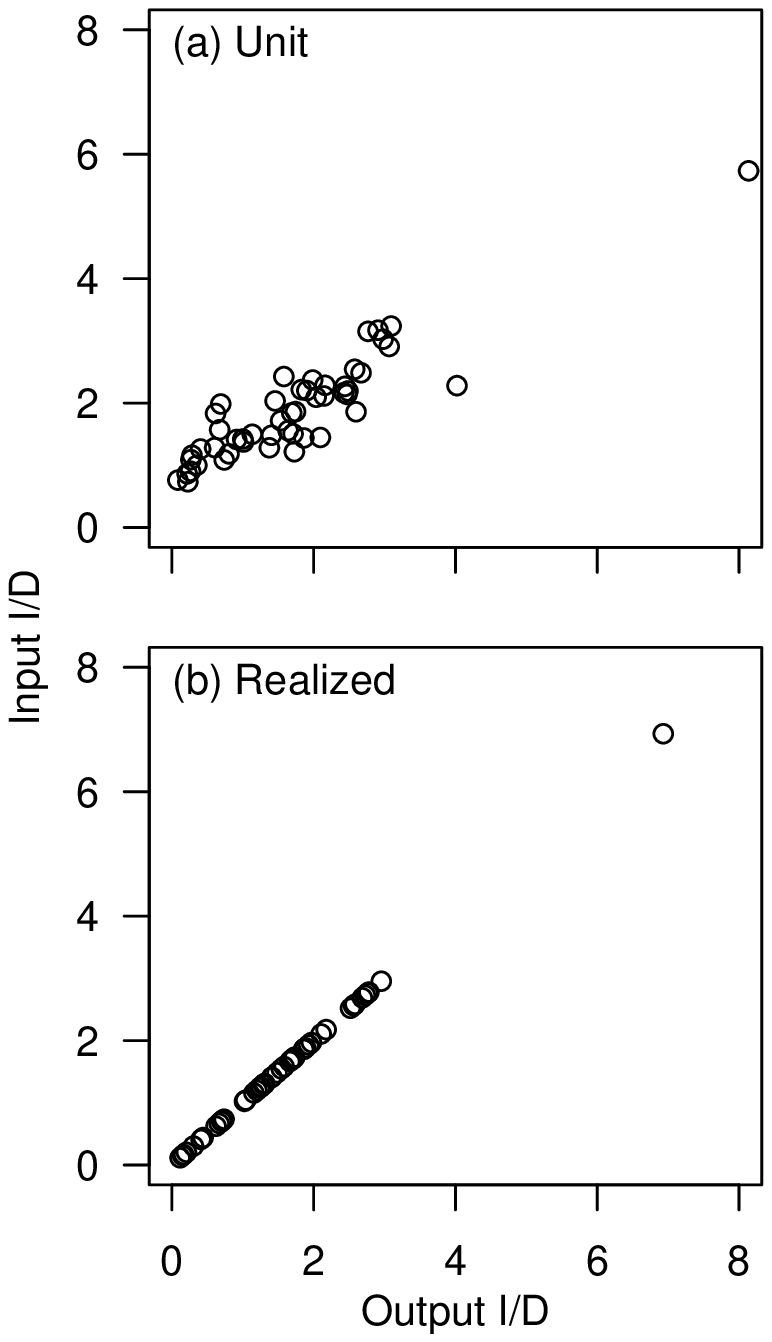}
 \caption{Relationship between input and output oriented measures of
   the ratio of indirect-to-direct effects calculated using the (a)
   unit and (b) realized metrics.} \label{fig:ide}
\end{figure}

\end{document}

%% file: table_models_ide.tex
\begin{table*}
\caption{Fifty trophically-based ecosystem network
  models.} \label{tab:models}
\tableline 
\begin{center}
\begin{footnotesize}
\begin{tabular}{l l c c r c r}
Model & units & $n^\dagger$ & $C^\dagger$ & $TST^\dagger$ &
$FCI^\dagger$ & Source \\
\hline \\[-1.5ex]
Lake Findley & gC m$^{-2}$ yr$^{-1}$ & 4 & 0.38 & 51 & 0.30 & \citet{richey78} \\
Mirror Lake & gC m$^{-2}$ yr$^{-1}$ & 5 & 0.36 & 218 & 0.32 &  \citet{richey78} \\
Lake Wingra & gC m$^{-2}$ yr$^{-1}$ & 5 & 0.40 & 1,517 & 0.40 & \citet{richey78} \\
Marion Lake & gC m$^{-2}$ yr$^{-1}$ & 5 & 0.36 & 243 & 0.31 & \citet{richey78} \\
Cone Springs & kcal m$^{-2}$ yr$^{-1}$ & 5 & 0.32 & 30,626 & 0.09 & \citet{tilly68} \\
Silver Springs & kcal m$^{-2}$ yr$^{-1}$ & 5 & 0.28 & 29,175 & 0.00 & \citet{odum57} \\
English Channel & kcal m$^{-2}$ yr$^{-1}$ & 6 & 0.25 & 2,280 & 0.00 & \citet{brylinsky72} \\
Oyster Reef & kcal m$^{-2}$ yr$^{-1}$ & 6 & 0.33 & 84 & 0.11 & \citet{dame81} \\
Somme Estuary & mgC m$^{-2}$ d$^{-1}$ & 9 & 0.30 & 2,035 & 0.14 & \citet{rybarczyk03} \\
Bothnian Bay & gC m$^{-2}$ yr$^{-1}$ & 12 & 0.22 & 130 & 0.18 &  \citet{sandberg00} \\
Bothnian Sea & gC m$^{-2}$ yr$^{-1}$ & 12 & 0.24 & 458 & 0.27 &  \citet{sandberg00} \\
Ythan Estuary & gC m$^{-2}$ yr$^{-1}$ & 13 & 0.23 & 4,181 & 0.24 & \citet{baird81} \\
Baltic Sea & mgC m$^{-2}$ d$^{-1}$ & 15 & 0.17 & 1,974 & 0.13 &  \citet{baird91} \\
Ems Estuary & mgC m$^{-2}$ d$^{-1}$ & 15 & 0.19 & 1,019 & 0.32 & \citet{baird91} \\
Swarkops Estuary & mgC m$^{-2}$ d$^{-1}$ & 15 & 0.17 & 13,996 & 0.47 &  \citet{baird91} \\
Southern Benguela Upwelling & mgC m$^{-2}$ d$^{-1}$ & 16 & 0.23 & 1,774 & 0.19 &\citet{baird91} \\
Peruvian Upwelling & mgC m$^{-2}$ d$^{-1}$ & 16 & 0.22 & 33,496 & 0.04 & \citet{baird91} \\
Crystal River (control) & mgC m$^{-2}$ d$^{-1}$ & 21 & 0.19 & 15,063 & 0.07 & \citet{ulanowicz86} \\
Crystal River (thermal) & mgC m$^{-2}$ d$^{-1}$ & 21 & 0.14 & 12,032 & 0.09 & \citet{ulanowicz86} \\
Charca de Maspalomas Lagoon & mgC m$^{-2}$ d$^{-1}$ & 21 & 0.13 & 6,010,331 & 0.18 & \citet{almunia99} \\
Northern Benguela Upwelling & mgC m$^{-2}$ d$^{-1}$ & 24 & 0.21 & 6,608 &0.05 & \citet{heymans00} \\
Neuse Estuary (early summer 1997) & mgC m$^{-2}$ d$^{-1}$ & 30 & 0.09 & 13,826 & 0.12 & \citet{baird04} \\
Neuse Estuary (late summer 1997) & mgC m$^{-2}$ d$^{-1}$ & 30 & 0.11 & 13,038 & 0.13 & \citet{baird04} \\
Neuse Estuary (early summer 1998) & mgC m$^{-2}$ d$^{-1}$ & 30 & 0.09 & 14,025 & 0.12 & \citet{baird04} \\
Neuse Estuary (late summer 1998) & mgC m$^{-2}$ d$^{-1}$ & 30 & 0.10 & 15,031 & 0.11 & \citet{baird04} \\
Gulf of Maine & g ww m$^{-2}$ yr$^{-1}$ & 31 & 0.35 & 18,382 & 0.15 &  \citet{link08} \\
Georges Bank & g ww m$^{-2}$ yr$^{-1}$ & 31 & 0.35 & 16,890 & 0.18 & \citet{link08} \\
Middle Atlantic Bight & g ww m$^{-2}$ yr$^{-1}$ & 32 & 0.37 & 17,917 & 0.18 & \citet{link08} \\
Narragansett Bay & mgC m$^{-2}$ yr$^{-1}$ & 32 & 0.15 & 3,917,246 & 0.51  & \citet{monaco97} \\
Southern New England Bight & g ww m$^{-2}$ yr$^{-1}$ & 33 & 0.03 & 17,597 & 0.16 & \citet{link08} \\
Chesapeake Bay  & mgC m$^{-2}$ yr$^{-1}$ & 36 & 0.09 & 3,227,453 & 0.19 & \citet{baird89} \\
St. Marks Seagrass, site 1 (Jan) & mgC m$^{-2}$ d$^{-1}$ & 51 & 0.08 & 1,316 & 0.13 & \citet{baird98} \\
St. Marks Seagrass, site 1 (Feb) & mgC m$^{-2}$ d$^{-1}$ & 51 & 0.08 & 1,591 & 0.11 & \citet{baird98} \\
St. Marks Seagrass, site 2 (Jan) & mgC m$^{-2}$ d$^{-1}$ & 51 & 0.07 & 1,383 & 0.09 & \citet{baird98} \\
St. Marks Seagrass, site 2 (Feb) & mgC m$^{-2}$ d$^{-1}$ & 51 & 0.08 & 1,921 & 0.08 & \citet{baird98} \\
St. Marks Seagrass, site 3 (Jan) & mgC m$^{-2}$ d$^{-1}$ & 51 & 0.05 & 12,651 & 0.01 &\citet{baird98} \\
St. Marks Seagrass, site 4 (Feb) & mgC m$^{-2}$ d$^{-1}$ & 51 & 0.08 & 2,865 & 0.04 & \citet{baird98} \\
Sylt R{\o}m{\o} Bight & mgC m$^{-2}$ d$^{-1}$ & 59 & 0.08 & 1,353,406 & 0.09 & \citet{baird04_sylt} \\
Graminoids (wet) & gC m$^{-2}$ yr$^{-1}$ & 66 & 0.18 & 13,677 & 0.02 & \citet{ulanowicz00_graminoids} \\
Graminoids (dry) & gC m$^{-2}$ yr$^{-1}$ & 66 & 0.18 & 7,520 & 0.04 &  \citet{ulanowicz00_graminoids} \\
Cypress (wet) & gC m$^{-2}$ yr$^{-1}$ & 68 & 0.12 & 2,572 & 0.04 & \citet{ulanowicz97_cypress} \\
Cypress (dry) & gC m$^{-2}$ yr$^{-1}$ & 68 & 0.12 & 1,918 & 0.04 & \citet{ulanowicz97_cypress}\\
Lake Oneida (pre-ZM) & gC m$^{-2}$ yr$^{-1}$ & 74 & 0.22 & 1,638 & $<0.01$ & \citet{miehls09_oneida} \\
Lake Quinte (pre-ZM) & gC m$^{-2}$ yr$^{-1}$ & 74 & 0.21 & 1,467 & $<0.01$ &  \citet{miehls09_quinte} \\
Lake Oneida (post-ZM) & gC m$^{-2}$ yr$^{-1}$ & 76 & 0.22 & 1,365 & $<0.01$ & \citet{miehls09_oneida} \\
Lake Quinte (post-ZM) & gC m$^{-2}$ yr$^{-1}$ & 80 & 0.21 & 1,925 & $0.01$ &  \citet{miehls09_quinte} \\
Mangroves (wet) & gC m$^{-2}$ yr$^{-1}$ & 94 & 0.15 & 3,272 & 0.10 & \citet{ulanowicz99_mangrove} \\
Mangroves (dry) & gC m$^{-2}$ yr$^{-1}$ & 94 & 0.15 & 3,266 & 0.10 & \citet{ulanowicz99_mangrove} \\
Florida Bay (wet) & mgC m$^{-2}$ yr$^{-1}$ & 125 & 0.12 & 2,721 & 0.14 & \citet{ulanowicz98_fb} \\
Florida Bay (dry) & mgC m$^{-2}$ yr$^{-1}$ & 125 & 0.13 & 1,779 & 0.08& \citet{ulanowicz98_fb} \\[0.5ex] 
\end{tabular}
\end{footnotesize}
\end{center}
\tableline
\begin{scriptsize}
$^\dagger$ $n$ is the number of nodes in the network model, $C=L/n^2$
is the model connectance when $L$ is the number of direct links or
energy--matter transfers, $TST=\sum\sum{f_{ij}}+\sum{z_i}$ is the total
system throughflow, and $FCI$ is the Finn Cycling Index \citep{finn80}.  
\end{scriptsize}
\end{table*}